\documentclass[conference]{IEEEtran}
\IEEEoverridecommandlockouts
\usepackage{cite}
\usepackage{amsthm}
\newtheorem{Theorem}{Theorem}
\usepackage{amsmath,amssymb,amsfonts}
\usepackage{algorithmic}
\usepackage{graphicx}
\usepackage{textcomp}
\usepackage{algorithm}
\usepackage{xcolor}

\newcommand\CH{\ensuremath{\bm{\mathcal{H}}}}   

\newcommand\Ac{\ensuremath{{\mathcal{A}}}}

\newcommand\Kc{\ensuremath{{\mathcal{K}}}}
\newcommand\Mc{\ensuremath{{\mathcal{M}}}}
\newcommand\Uc{\ensuremath{{\mathcal{U}}}}

\newcommand\wb{\ensuremath{{\bf w}}}

\newcommand\ub{\ensuremath{{\bf u}}}

\newcommand\Ib{\ensuremath{{\bf I}}}

\newcommand\Vb{\ensuremath{{\bf V}}}

\graphicspath{{figures/}}

\definecolor{applegreen}{rgb}{0.55, 0.71, 0.0}

\newtheorem{assumption}{Assumption}

	\definecolor{black}{rgb}{0.0, 0.0, 0.0}
	\definecolor{white}{rgb}{1.0, 1.0, 1.0}
	\definecolor{yellow}{rgb}{1.0, 1.0, 0.8}
	\definecolor{red}{rgb}{0.6, 0.0, 0.2}
	\definecolor{blue}{rgb}{0.0, 0.2, 0.5}
	\definecolor{green}{rgb}{0.6, 0.8, 0.8}
	\definecolor{dark_green}{RGB} {0, 140, 0}
	\definecolor{gold}{rgb}{0.6, 0.4, 0.1}
	\definecolor{grey}{RGB}{0,0,0}
	\definecolor{Gray}{gray}{0.8}
	\definecolor{MediumGray}{gray}{0.9}
	\definecolor{LightGray}{gray}{0.98}
	\definecolor{LightCyan}{rgb}{0.88,1,1}
	\definecolor{purple}{RGB}{128,0,128}
	\definecolor{sl_blue}{RGB}{47, 60, 105}
	\definecolor{orange}{RGB}{255,165,0}
	\definecolor{Gray}{gray}{0.85}
	
	\newif\ifcomments
	\commentstrue

	\ifcomments
	\newcommand{\JUN}[1]{\textcolor{orange}{\textbf{\small [}\colorbox{yellow}{\textbf{Jun:}}{\small #1}\textbf{\small ]}}}
	\newcommand{\MB}[1]{\textcolor{dark_green}{\textbf{\small [}\colorbox{yellow}{\textbf{Mauro:}}{\small #1}\textbf{\small ]}}}
	\newcommand{\BT}[1]{\textcolor{brown}{{#1}}}
	\newcommand{\TODO}[1]{\textcolor{red}{{TODO: #1}}}
	\else
	\newcommand{\JUN}[1]{}
	\newcommand{\MB}[1]{}
	\newcommand{\BT}[1]{}
	\newcommand{\TODO}[1]{}
	\newcommand{\CH}[1]{}
	\fi
\def\BibTeX{{\rm B\kern-.05em{\sc i\kern-.025em b}\kern-.08em
    T\kern-.1667em\lower.7ex\hbox{E}\kern-.125emX}}
\begin{document}

\title{QoS-Constrained Scheduling in Multi-Cell Multi-User MIMO Networks\\
\thanks{The work is supported by National Key R\&D Program of China under grant 2024YFA1014204, by Guangdong Provincial Key Laboratory of Big Data Computing and by the NSFC, China, under Grant 62501513.}
}

\author{
\IEEEauthorblockN{
Tenghao Cai\textsuperscript{1,3},
Lei Li\textsuperscript{2,3},
and Tsung-Hui Chang\textsuperscript{2,3}
}
\IEEEauthorblockA{
\textsuperscript{1}\textit{School of Science and Engineering, The Chinese University of Hong Kong, Shenzhen} \\
\textsuperscript{2}\textit{School of Artificial Intelligence, The Chinese University of Hong Kong, Shenzhen} \\
\textsuperscript{3}\textit{Shenzhen Research Institute of Big Data} \\
Shenzhen, China \\
221019048@link.cuhk.edu.cn, lei.ap@outlook.com, changtsunghui@cuhk.edu.cn
}
}

\maketitle

\begin{abstract}
    In 5G and beyond networks, efficient scheduling is essential to exploit the gains of multi-user MIMO (MU-MIMO) equipped with carrier aggregation and joint transmission (JT). However, cross-cell and cross-carrier scheduling under QoS constraints is challenging due to the strong coupling across users, base stations, and carriers. In this work, we address this problem in multi-cell MU-MIMO networks to maximize system throughput for both JT and non-JT users under rate constraints. The optimization is highly complex as scheduling variables and beamforming (BF) vectors are intertwined. To tackle it, we propose an approximate but tractable surrogate by leveraging the eigen-based zero-forcing BF  and massive MIMO asymptotics. The reformulated problem has a separable structure and is amenable to efficient solutions by a penalty-based block coordinate descent method. Simulations demonstrate that the proposed scheduler not only meets the QoS requirements well but also achieves remarkable throughput gains over existing schemes.
\end{abstract}

\begin{IEEEkeywords}
5G, multi-cell MU-MIMO, joint transmission, scheduling, QoS constraints.
\end{IEEEkeywords}

\section{Introduction}
     To meet the rapid diversification of mobile services and the escalating demand for quality of service (QoS), multi-user multiple-input multiple-output (MU-MIMO), carrier aggregation (CA), and joint transmission (JT) have become key enabling technologies for 5G and beyond. MU-MIMO allows base stations (BSs) to simultaneously serve multiple user equipments (UEs) within the same time–frequency resources via spatial multiplexing; CA enhances capacity by enabling UEs to access multiple component carriers (CCs) at different center frequencies, each divided into resource block groups (RBGs) for fine-grained spectrum allocation; and JT improves signal quality and coverage for cell-edge UEs through coherent transmission from multiple BSs. While these techniques significantly improve network performance, fully realizing their potential requires efficient scheduling strategies. In particular, cross-cell and cross-carrier scheduling in MU-MIMO networks with JT is highly challenging due to the strong coupling of variables across CCs, UEs and BSs.

    Multi-cell MIMO scheduling has been extensively studied in the literature. Leveraging fractional programming and the Hungarian method, a multi-band scheduling approach was proposed in \cite{khan2020optimizing}, and  GPU-based schedulers were devised in \cite{chen2022m, 9488684} to achieve real-time processing. Although effective, these designs overlooked user QoS requirements, which, if incorporated, introduce highly non-convex constraints and significantly complicate the problem.
    In view of this, \cite{yu2012coordinated} proposed a branch-and-bound-based algorithm for subchannel scheduling and beamforming (BF) optimization under QoS constraints to maximize the number of scheduled UEs. Besides, \cite{he2022cross} developed an iterative algorithm alternating between BF optimization and user scheduling to maximize the sum rate. However, these approaches require iterative BF by solving either a second-order cone program \cite{yu2012coordinated} or a minimum mean squared error (MSE) problem \cite{he2022cross}. The resulting complexity remains high and unaffordable to existing systems. 
    
    In industrial practice, structured BF strategies, like zero-forcing (ZF) \cite{li2015multicell} or eigen-based ZF (EZF) \cite{zhao2023communication} are more appealing, as they can be obtained in closed-form at much lower complexity once scheduling decisions are made. Unfortunately, even with ZF/EZF-BF, the scheduling design remains challenging due to their strong coupling. To handle this, several heuristic schemes have been proposed \cite{yoo2006optimality, idrees2022throughput}. A semi-orthogonal user selection (SUS) with ZF was devised in \cite{yoo2006optimality}, where each BS greedily selects users with high channel gains and nearly orthogonal channels. However, this approach fails to account for inter-cell interference (ICI), leading to degraded performance in multi-cell scenarios. Applying multi-cell cooperative ZF for cell-edge users, the SHS scheduler \cite{idrees2022throughput} sequentially selects one user per resource block (RB) and dynamically excludes users that have already satisfied their QoS requirements. Nevertheless, limiting each RB to a single user greatly reduces the resource utilization efficiency.

    In this work, we investigate the cross-cell and cross-carrier scheduling in MU-MIMO networks with both JT and non-JT (NJT) UEs. Applying the practical EZF-BF, we formulate a scheduling optimization problem to guarantee the QoS of specific users while maximizing the sum rate of the others. Since the values of BF rely on the scheduling variables during optimization, the strong variable coupling makes the problem complex to solve. To handle it, we first propose a novel reformulation scheme that exploits EZF and massive-MIMO asymptotic properties to eliminate the intermediate BF variables, attaining simplified rate expressions with a separable structure. Then, we introduce a penalty-based block coordinate descent (BCD) algorithm to handle the surrogate problem efficiently. Extensive results demonstrate that the proposed scheduling algorithm not only ensures QoS satisfaction but also achieves significant throughput gains over existing schedulers.

\section{System Model \& Problem Formulation}

    \begin{figure}[t] 
        \centering	
        {\includegraphics[width=0.48\textwidth]{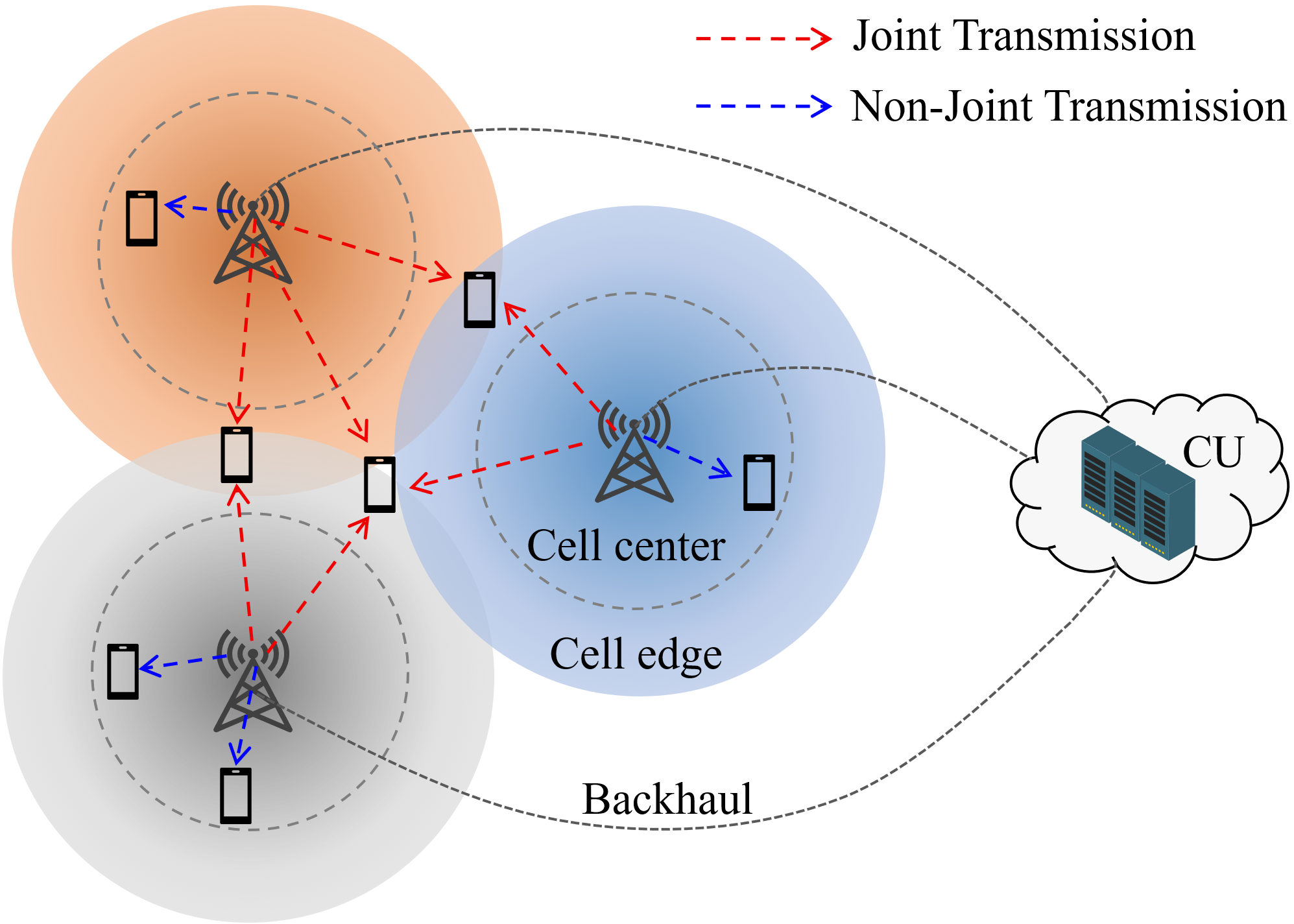}}
        \caption{System model.} 
        \label{fig:sys} 
    \end{figure}
\subsection{System Model}
     As shown in Fig.~\ref{fig:sys}, we consider a wireless network including $M$ cells and $K$ users, where each cell center deploys a BS equipped with $N_t$ antennas, and all BSs are connected to a central unit (CU). Each UE is equipped with $N_r$ antennas. Denote $\Mc$ and $\Kc$ as the sets of BSs and UEs, respectively. The set of UEs associated with BS $m$ is denoted as $\Ac_m$, which is assumed to be known and partitioned as $\Ac_m = \{\Kc_m, \Uc_m\}$, with $\Uc_m$ being the set of cell-edge UEs experiencing strong ICI and $\Kc_m$ being that of cell-center UEs mainly affected by intra-cell interference. In the network, neighbouring BSs employ JT to serve the cell-edge users (termed JT-UEs) to improve their performance, while cell-center UEs are served solely by their own BSs and are termed NJT-UEs. Denote the set of serving BSs for UE $k$ as $\mathcal{B}_k$. 
    The overall frequency spectrum of the system spans over $C$ CCs indexed by the set $\mathcal{C}$, and each CC consists of $R$ RBGs indexed by the set $\mathcal{R}$, so that each RBG is identified by a pair $(c,r)$ with $c \in \mathcal{C}$  and $r \in \mathcal{R}$. Denote $b_{k}^{c,r}$ as the scheduling variable indicating whether UE $k \in \mathcal{A}_m$ is allocated to RBG $r$ of the $c$-th CC, and
    \begin{align}
        b_{k}^{c,r}= \begin{cases}1, & \text{if UE $k$ is scheduled      on RBG $(c,r)$,} \\ 0, & \text{otherwise.}\end{cases}
    \end{align}
    
    Assume that channel state information (CSI) $\mathbf{H}_{m,k}^{c,r} \in \mathbb{C}^{N_r \times N_t}$ from BS $m$ to UE $k$ on RBG $(c,r)$ is available.  The transmit data symbol for UE $k$ satisfies $\mathbb{E}[|x_k|^2] = 1$ and is assumed to be mutually independent across different UEs. After applying the BF $\wb^{c,r}_{m,k}$ and the receive combiner $\ub_k^{c,r}$, the received signal of UE $k$ on RBG $(c,r)$ is written as  
    \begin{equation} \label{eq:sig_rx_NJT}
    \begin{aligned} 
        {y}_{k}^{c,r} 
         = &\sum\nolimits_{m \in \mathcal{B}_k} b_{k}^{c,r} (\mathbf{u}^{c,r}_{k})^\mathrm{H} \mathbf{H}_{m,k}^{c,r} \mathbf{w}_{m,k}^{c,r}x_k  \\
        &+  \sum\nolimits_{ j \neq k}  (\mathbf{u}^{c,r}_{k})^\mathrm{H}  \sum\nolimits_{\ell \in \mathcal{B}_j} b_{j}^{c,r}\mathbf{H}_{\ell,k}^{c,r} \mathbf{w}_{\ell,j}^{c,r}x_j\\
        &+ (\mathbf{u}^{c,r}_{k})^\mathrm{H}\mathbf{n}_k, 
    \end{aligned}
    \end{equation}
    where $\mathbf{n}_k \sim \mathcal{CN}(0, \sigma^2 \Ib_{N_r})$ is the additive complex Gaussian noise.  The first terms in \eqref{eq:sig_rx_NJT} represent the desired signal for UE $k$, and the second term represents the interference from co-scheduled UEs. 
    Then, the signal-to-interference-plus-noise-ratio (SINR) of UE $k$ on RBG $(c,r)$ is
    \begin{equation} \label{gamma}
    \begin{aligned}
        \gamma_{k}^{c,r}
        = &  \left| \sum\nolimits_{m \in \mathcal{B}_k} b_{k}^{c,r} (\mathbf{u}^{c,r}_{k})^\mathrm{H} \mathbf{H}_{m,k}^{c,r} \mathbf{w}_{m,k}^{c,r} \right|^2  \\
        & \!\!\!\!\!\! \times   \left(  \sum\nolimits_{ j \neq k}  \big|  (\mathbf{u}^{c,r}_{k})^\mathrm{H}  \sum\nolimits_{\ell \in \mathcal{B}_j} b_{j}^{c,r}\mathbf{H}_{\ell,k}^{c,r} \mathbf{w}_{\ell,j}^{c,r} \big|^2 
         + \sigma^2 \right)^{-1}, \notag
    \end{aligned}  
    \end{equation}
    Accordingly, the achievable data rate is
    \begin{align}\label{realrate}
        f_{k}^{c,r} = \log(1 + b_{k}^{c,r}\gamma_{k}^{c,r}), \forall c\in \mathcal{C}, r\in \mathcal{R}, k\in \mathcal{K}.
    \end{align}
    
\subsection{Transceiver Design}
    Given the computational constraints and real-time processing requirements, BF algorithms with closed-form/semi-closed-form solutions are strongly preferred over iterative ones in industrial practice. This motivates us to formulate the scheduling problem in this work using EZF-BF. Compared to standard ZF-BF, EZF-BF utilizes eigen-decomposition to operate on dominant channel eigenmodes, improving the robustness in scenarios with ill-conditioned channels and achieving more stable performance.
    
    Specifically, for each
    NJT-UE $k \in \mathcal{K}_m$, given $\mathbf{H}_{m,k}^{c,r}$, BS $m$ performs SVD 
    \begin{align}
         \mathbf{H}_{m,k}^{c,r}=\mathbf{T}_{k}^{c,r}\mathbf{\Lambda}_{k}^{c,r}(\mathbf{V}_{m,k}^{c,r})^\mathrm{H},
    \end{align}
    where the diagonal $\mathbf{\Lambda}_{k}^{c,r}$ contains the singular values in descending order (largest $\lambda_{k}^{c,r}$), and $\mathbf{T}_{k}^{c,r}$, $\mathbf{V}_{m,k}^{c,r}$ contain the left and right singular vectors, respectively. 
    In contrast, for JT-UE $i  \in \mathcal{U}_m$ served by multiple BSs, the transceiver is jointly designed based on all channels from its serving BSs.  
    Specifically, by stacking all the channel matrices from the serving BSs $m \in \mathcal{B}_i$ to JT-UE $i$, its aggregated channel matrix on RBG $(c,r)$ is 
    \begin{align} \label{eq:SVd_JT}
        \mathbf{H}_{i}^{c,r} \triangleq   \left[ \mathbf{H}_{m,i}^{c,r} \right]_{m \in \mathcal{B}_i}  \in \mathbb{C}^{N_r \times |\mathcal{B}_i|N_t }, 
    \end{align}
    whose SVD is
    \begin{align}\label{JT_Right}
    \mathbf{H}_{i}^{c,r} = \mathbf{T}_{i}^{c,r} \mathbf{\Sigma}_{i}^{c,r} (\mathbf{V}_{i}^{c,r})^\mathrm{H},
    \end{align}
    where $\mathbf{V}_{i}^{c,r}\in \mathbb{C}^{|\mathcal{B}_i|N_t \times |\mathcal{B}_i|N_t}$ is the unitary right singular vector matrix. $\mathbf{\Sigma}_{i}^{c,r} \in \mathbb{C}^{N_r \times |\mathcal{B}_i|N_t}$ holds the singular values in descending order, with its largest one $\lambda_{i}^{c,r}$. 
    It can be found from \eqref{eq:SVd_JT} that $\{\mathbf{H}_{m,i}^{c,r}\}_{m \in \mathcal{B}_i}$ share the same left singular vectors and singular values, i.e., 
    \begin{align}
        \mathbf{H}_{m,i}^{c,r} = \mathbf{T}_i^{c,r} \mathbf{\Sigma}_i^{c,r} (\mathbf{V}_{m,i}^{c,r})^\mathrm{H},
    \end{align} 
    where $\mathbf{V}_{m,i}^{c,r}$ is the corresponding matrices that can be attained from $\Vb_i^{c,r}$ in \eqref{JT_Right}.
    
    Let $\mathbf{t}_{k}^{c,r}$ and $\mathbf{v}_{m,k}^{c,r}$ be the first column of $\mathbf{T}_{k}^{c,r}$ and $\mathbf{V}_{m,k}^{c,r}$, respectively. To maximize the receive signal-to-noise-ratio (SNR), we set $\mathbf{u}_{k}^{c,r} = \mathbf{t}_{k}^{c,r}$ and it leads to
    \begin{align}\label{uH}   (\mathbf{u}_{k}^{c,r})^{\mathrm{H}}\mathbf{H}_{m,k}^{c,r} = \lambda_k^{c,r} (\mathbf{v}_{m,k}^{c,r})^{\mathrm{H}}, \quad \forall c,r,m,k.
    \end{align} 
    
    Denote NJT-UEs in cell $m$ scheduled on RBG $(c,r)$ as $\mathcal{K}_m^{c,r}\! \triangleq \! \{k_1,\dots,k_m\} \! \subseteq \! \mathcal{K}_m$, and JT-UEs scheduled on the same RBG as $\mathcal{U}_m^{c,r} \! \triangleq \! \{i_1,\dots,i_m\} \! \subseteq  \! \mathcal{U}_m$. By stacking their right singular vectors into $\hat{\mathbf{V}}_m^{c,r} \! \triangleq \! [{\mathbf{v}}_{m,k_1}^{c,r}, \dots, {\mathbf{v}}_{m,i_m}^{c,r}]$, the EZF-BF at BS $m$ for RBG $(c,r)$ is constructed as
    \begin{align}\label{bF}
        \hat{\mathbf{W}}_m^{c,r} = &\hat{\mathbf{V}}_m^{c,r} ( (\hat{\mathbf{V}}_m ^{c,r})^\mathrm{H} \hat{\mathbf{V}}_m^{c,r} )^{-1} 
        =  [ \hat{\mathbf{w}}_{m,k_1}^{c,r}, \ldots, \hat{\mathbf{w}}_{m,i_{m}}^{c,r}  ].
    \end{align} 
     Denote the power budget of BS $m$ per RBG as $P_m$, and 
     \begin{align}
         A_m^{c,r} \triangleq \sum\nolimits_{z \in \mathcal{A}_m} b_{z}^{c,r}
     \end{align} 
     as the number of UEs scheduled on RBG $(c,r)$, the EZF-BF for UE $j \in (\mathcal{K}_m^{c,r} \cup \mathcal{U}_m^{c,r})$ under In view of the industrial practice, let us formulate the scheduling problem under EZF-BF \cite{sun2010eigen} and equal power allocation is given by
     \begin{equation}\label{normalized}
        \mathbf{w}_{m,j}^{c,r} = \sqrt{{P_m}/{A_m^{c,r}} }\hat{\mathbf{w}}_{m,j}^{c,r}/\| \hat{\mathbf{w}}_{m,j}^{c,r} \|. 
     \end{equation}
\subsection{Problem Formulation}
    In practice, UEs of the network request diverse services and have different QoS requirements. Without loss of generality (w.l.o.g), we assume that UEs in the subset $\mathcal{Q} \subseteq \mathcal{K}$ have explicit rate constraints, while the set of the other UEs is denoted as $\hat{\mathcal{Q}} \triangleq \mathcal{K} \setminus \mathcal{Q}$. In this work, the scheduling is optimized to maximize the sum rate of the UEs $k \in \hat{\mathcal{Q}}$ while ensuring the QoS requirements of the UEs $j \in \mathcal{Q}$. Accordingly, the user scheduling optimization can be formulated as
    \begin{subequations}\label{SRmax-original}
        \begin{align}
            \max_{\{b^{c,r}_{k}\}}  & \sum\nolimits_{(c,r)} \sum\nolimits_{k \in \mathcal{\hat{Q}}}f_{ k}^{c,r} \label{SRmax0_a}\\
            \text { s.t. } 
            & b^{c,r}_{k}\in\{0,1\}, \forall c\in \mathcal{C}, r\in \mathcal{R}, k\in \mathcal{K}\\
            &\sum\nolimits_{(c,r)} f_{j}^{c,r} \geq Q_{j}, \forall  j  \in  \mathcal{Q}.  \label{KNF0}
        \end{align}
     \end{subequations} 
     Problem \eqref{SRmax-original} is a nonlinear integer programming whose solution space scales exponentially as $\mathcal{O}(2^{K \times C \times R})$ as users and RBGs increase. Besides the scheduling variables $\{b_{k}^{c,r}\}$ to optimize, problem \eqref{SRmax-original} contains the intermediate variables $\{\wb_{k}^{c,r}\}$ dependent $\{b_{k}^{c,r}\}$, as shown from \eqref{bF} and \eqref{normalized}. This strong coupling results in complex rate expressions $\{f^{c,r}_k \}$, rendering this problem neither convex nor concave. Moreover, the scheduling across multi-cell multi-CC involves both JT and NJT, as well as UEs with diverse QoS requirements. These factors pose a great challenge to the design of polynomial-time algorithms capable of achieving even a local optimum.

		
		

 \section{A Novel Problem Reformulation}\label{sec:reformu}  
    To solve problem \eqref{SRmax-original}, we first propose a novel reformulation that is amenable to efficient solutions. In MU-MIMO networks, allocating an RBG to UEs with weak inter-user interference is advantageous, since this typically yields a high SINR \cite{yoo2006optimality}. Motivated by this, we approximate the rate for UE $k \in \mathcal{K}$ as 
    \begin{subequations}
        \begin{align}
        {f}_{k}^{c,r} & = \log(1+ {b}_{k}^{c,r}{\gamma}_{k}^{c,r}), \\
        &\approx \log(b_{k}^{c,r}\gamma_{k}^{c,r}) \triangleq \tilde{f}_{k}^{c,r}.
    \end{align}
    \end{subequations}
    In addition, a key observation is that cell-center UEs usually receive a much stronger desired signal from their serving BS compared to the ICI, while cell-edge UEs experience comparable signal strengths from multiple BSs. 
    This motivates us to apply NJT for the former and JT for the latter. 
     Accordingly, it is reasonable to make the following assumption:
    \begin{assumption} \label{assumption}
		The NJT-UEs do not suffer from ICI, and the JT-UEs are not interfered by non-serving BSs \cite{chen2022m}.
	\end{assumption}
    {Based on Assumption \ref{assumption} and  \eqref{uH}-\eqref{normalized}, it holds that 
   \begin{align}
        (\mathbf{v}_{m,k}^{c,r})^\mathrm{H}\mathbf{w}_{m,j}^{c,r}
        &=
        \begin{cases}
         {\sqrt{P_m/A_m^{c,r}}}/{\| \hat{\mathbf{w}}_{m,k}^{c,r} \|},
        & j = k, \\[6pt]
        0, & j \neq k .
        \end{cases}
    \end{align} 
    and thereby the SINR of UE $k \in \mathcal{K}$ in \eqref{gamma} reduces to
    \begin{align}\label{JT SINR}
        {\gamma}_{k}^{c,r} =    \bigg|\sum\nolimits_{m \in \mathcal{B}_k}{\lambda_{k}^{c,r} \sqrt{P_m/A_m^{c,r}}}/{ \|\hat{\mathbf{w}}_{m,k}^{c,r}\| } \bigg|^2  /{\sigma^2}.
    \end{align}
    Note that $\|\hat{\mathbf{w}}_{m,k}^{c,r} \|$ depends on the scheduling variables of BS $m \in \mathcal{B}_k$ inexplicitly, as shown in \eqref{bF}. Moreover, for JT-UE $k$, \eqref{JT SINR} entails cross-BS coupling (i.e. $| \mathcal{B}_k| > 1$) and is more complicated than that for NJT-UEs (i.e. $| \mathcal{B}_k| = 1$). 
    
    Observing that all the terms within the summation of \eqref{JT SINR} are positive, we apply the inequality $|\sum_m a_m|^2 \geq \sum_m |a_m|^2$ \cite{bian2024decentralizing} to attain a lower bound with a separate structure as
    \begin{align} \label{eq:gamma_JT_s}
        \tilde{\gamma}_{k}^{c,r} \triangleq \sum\nolimits_{m \in \mathcal{B}_i}(\lambda_{k}^{c,r})^2P_m/(A_m^{c,r} \|\hat{\mathbf{w}}_{m,k}^{c,r}\|^2\sigma^2 ).   
    \end{align}
    Now, each term in \eqref{eq:gamma_JT_s} can be characterized uniformly using a per-BS component defined as 
    \begin{align}\label{gamma_m}
        \gamma_{m,k}^{c,r} \triangleq {(\lambda_{k}^{c,r} )^2P_m}/(A_m^{c,r}{\sigma^2 \|\hat{\mathbf{w}}_{m,k}^{c,r}\|^2}),
    \end{align}
    which represents the SINR contribution from BS $m$ to UE $k$. 
    Hence, the SINR for NJT-UE $k$ served by BS $m$ is $\gamma_{k}^{c,r} = \gamma_{m,k}^{c,r}$, while the lower-bound SINR for a JT-UE $k$ becomes  
    $\tilde{\gamma}_{k}^{c,r} \triangleq \sum\nolimits_{m \in \mathcal{B}_i} \gamma_{m,k}^{c,r}$.

    Although the SINR expressions $\{\gamma^{c,r}_{m,k}\}$ have been effectively simplified, the BFs $\{\hat{\mathbf{w}}_{m,k}^{c,r}\} $ in them are still strongly coupled with the scheduling variables. To handle it, we leverage a large antenna array and have the following theorem.}
    \begin{Theorem} \label{Th1}
        Given a sufficiently large $N_t$, the approximate rate of NJT-UE $k$ served by BS $m$ on RBG $(c,r)$ converges:
        \begin{align}\label{JTSINR1}
            \tilde{f}_{k}^{c,r} \approx b_{k}^{c,r}(\psi_{m, k}^{c,r}\! +\!\sum\nolimits_{j \in \!\mathcal{A}_m\setminus\{k\}} \! b_{ j}^{c,r} d_{m, j, k}^{c,r} - g_m^{c,r} ),   
        \end{align}
        where ${\psi}_{m,k}^{c,r}  \triangleq       \log({{({\lambda}}_{k}^{c,r})^2|\mathbf{v}_{m,k}^{c,r} |^2 P_m}/{\sigma^2})$, ${d}_{m,j,k}^{c,r} \triangleq  \log( 1 -  \eta_{m,j,k}^{c,r})$, $\eta_{m,j,k}^{c,r} \triangleq { | (\mathbf{v}_{m,j}^{c,r}) ^\mathrm{H} \mathbf{v}_{m,k}^{c,r} |^2}/({|\mathbf{v}_{m,j}^{c,r}|^2|\mathbf{v}_{m,k}^{c,r}|^2})$
          and $g_m^{c,r} \triangleq \log(A_m^{c,r}) $.
    \end{Theorem}
    \begin{proof}
        As the derivations of the SINR on different RBGs are similar, we omit the superscript $(c,r)$ in the following for ease of illustration.  
	       First, based on the generation of EZF-BF in \eqref{bF}, we have $\hat{\mathbf{W}}_m^\mathrm{H} \hat{\mathbf{W}}_m = (  \hat{\mathbf{V}}_m ^\mathrm{H} \hat{\mathbf{V}}_m )^{-1}$. W.l.o.g., assume the BF $\hat \wb_{m,k}$ of the UE $k$ takes the $j$-th column of $\hat{\mathbf{W}}_m$ in \eqref{bF}, then
    \begin{equation}\label{eq:ww}
        \|\hat{\mathbf{w}}_{m,k} \|_2^2 = (  \hat{\mathbf{V}}_m ^\mathrm{H} \hat{\mathbf{V}}_m)_{j,j}^{-1}.
    \end{equation}
	Denote  $\mathbf{\hat{V}}_{m,-j}$ as the sub-matrix of $\mathbf{\hat{V}}_m$ without the $j$-th column $\mathbf{v}_{m,k}$. By adopting the block matrix inversion formula \cite{li2015multicell}, $ (\hat{\mathbf{V}}_m ^\mathrm{H} \hat{\mathbf{V}}_m)_{j,j}^{-1}$ is further expanded as
    \begin{align} \label{eq:invV}
        \!\!\!\!\frac{1}{\mathbf{v}_{m,k} ^\mathrm{H}\mathbf{v}_{m,k} \!- \!\mathbf{v}_{m,k} ^\mathrm{H} \hat{\mathbf{V}}_{m,-j} ( \hat{\mathbf{V}}_{m,-j}^\mathrm{H} \hat{\mathbf{V}}_{m,-j} )^{-1} \hat{\mathbf{V}}_{m,-j} ^\mathrm{H} \mathbf{v}_{m,k}}. 
    \end{align}
	Combining \eqref{eq:invV} with \eqref{eq:ww},  $1/{\|\hat{\mathbf{w}}_{m,k} \|_2^2}$ equals
	\begin{align}\label{eq:inv_ww}
			\!\!\!\! \mathbf{v}_{m,k} ^\mathrm{H}\mathbf{v}_{m,k} \!- \!\mathbf{v}_{m,k} ^\mathrm{H} \hat{\mathbf{V}}_{m,-j} ( \hat{\mathbf{V}}_{m,-j}^\mathrm{H} \hat{\mathbf{V}}_{m,-j} )^{-1} \hat{\mathbf{V}}_{m,-j} ^\mathrm{H} \mathbf{v}_{m,k}.
	\end{align}
	Moreover, it can be observed from \eqref{eq:inv_ww} that the diagonal entries of $ \hat{\mathbf{V}}_{m,-j} ^\mathrm{H} \hat{\mathbf{V}}_{m,-j} $ equal $\mathbf{v}_{m,z} ^\mathrm{H}\mathbf{v}_{m,z}, z \neq k$, while the non-diagonal entries correspond to the inter-product of the right singular vectors between different UEs on the same RBG. According to the asymptotics of massive MIMO \cite{vershynin2009high}, when the number of transmitting antennas tends to infinity, $\mathbf{v}_{m,z} ^\mathrm{H}\mathbf{v}_{m,k} \rightarrow 0$ holds for any two users $z\neq k$ with different channels. This implies $ \hat{\mathbf{V}}_{m,-j} ^\mathrm{H} \hat{\mathbf{V}}_{m,-j}  $ is close to a diagonal matrix for a large $N_t$, and \eqref{eq:inv_ww} can be approximated as 
	\begin{align} \label{eq:app_diag}
			& \mathbf{v}_{m,k} ^\mathrm{H}\mathbf{v}_{m,k}  - \sum\nolimits_{\substack{z \in \mathcal{A}_m, z \neq k}} {b_{z} | \mathbf{v}_{m,z} ^\mathrm{H} \mathbf{v}_{m,k} |^2}/{|\mathbf{v}_{m,z}|^2}  \\
			& = |\mathbf{v}_{m,k}|^2\bigg(1 \! -\!\!\! \sum_{\substack{z \in \mathcal{A}_m, z \neq k}} \!\! {b_{z} | \mathbf{v}_{m,z} ^\mathrm{H} \mathbf{v}_{m,k} |^2}/({|\mathbf{v}_{m,z}|^2|\mathbf{v}_{m,k}|^2})\bigg).\notag
	\end{align}
	Moreover, as $1-\sum_{i=1}^n x_i \approx \prod_{i=1}^n\left(1-x_i\right)$ holds when $x_i \approx 0$, we have
	\begin{align} \label{eq:prod}
			&1  - \sum\nolimits_{ {z \in \mathcal{A}_m, z \neq k}}  {b_{z} | \mathbf{v}_{m,z} ^\mathrm{H} \mathbf{v}_{m,k} |^2}/({|\mathbf{v}_{m,z}|^2|\mathbf{v}_{m,k}|^2})  \\
			&\approx \prod\nolimits_{ {z \in \mathcal{A}_m, z \neq k}} \left( 1 \!-\!  {b_{z} | \mathbf{v}_{m,z} ^\mathrm{H} \mathbf{v}_{m,k} |^2}/({|\mathbf{v}_{m,z}|^2|\mathbf{v}_{m,k}|^2} )\right)  \notag.
	\end{align}
	Substituting \eqref{eq:prod} and \eqref{eq:app_diag}  into \eqref{gamma_m} obtains the \eqref{JTSINR1}.
    \end{proof}

        

    Theorem \ref{Th1} successfully removes the intermediate BF variables and the attained $\tilde{f}_{k}^{c,r}$ solely relies on the scheduling variables in the serving cell $m$; therefore, we introduce  $\tilde{f}_{m,k}^{c,r}\triangleq \tilde{f}_{k}^{c,r}$. Moreover, it reveals the interplay among the scheduling variables of different UEs clearly. From \eqref{JTSINR1}, we observe that the rate is influenced by the following three components:
	\begin{enumerate}
		\item ${\psi}_{m,k}^{c,r}$, corresponds to the rate contribution from the SNR of UE $k$ when scheduled exclusively on RBG $(c,r)$;
		
		\item $\sum\nolimits_{j \in \mathcal{A}_m\setminus\{k\}} \! b_{j}^{c,r} d_{m, j, k}^{c,r}$, where $d_{m, j, k}^{c,r}$ is negative, quantifies the loss due to interference from other UEs with correlated channels (i.e., co-channel interference);
		
		\item $g_{m}^{c,r}$, captures the impact of power sharing among co-scheduled UEs.
	\end{enumerate}
    
    For JT-UE $i \in \mathcal{U}_m$, similarly we have
    \begin{align}
    \tilde{f}_{i}^{c,r} = \log\big(b_{i}^{c,r}  \tilde{\gamma}_{i}^{c,r}\big) = \log(b_{i}^{c,r}\sum\nolimits_{m\in \mathcal{B}_i} {\gamma}_{m,i}^{c,r}),
    \end{align}
    where the logarithmic function couples the SINR terms from multiple BSs, making the expression non-separable across the serving BSs. To handle it, we further apply Jensen’s inequality to obtain a separable approximate lower bound of $\tilde{f}_i^{c,r}$ as 
    \begin{align} \label{rate1}
        &\sum\nolimits_{m\in \mathcal{B}_i}\tilde{f}_{m,i}^{c,r} \triangleq \sum\nolimits_{m\in \mathcal{B}_i}\log(b_{i}^{c,r}|\mathcal{B}_i|{\gamma}_{m,i}^{c,r})/|\mathcal{B}_i|,  \\
        &= \sum_{m\in \mathcal{B}_i}b_{i}^{c,r}(\tilde{\psi}_{m, i}^{c,r}\! +\!\sum\nolimits_{j \in \!\mathcal{A}_m\setminus\{i\}} \! b_{j}^{c,r} d_{m, j, i}^{c,r} - g_m^{c,r} )/|\mathcal{B}_i|,  \nonumber
    \end{align}
    where $\tilde{\psi}_{m,i}^{c,r} = \log(|\mathcal{B}_i|)+{\psi}_{m,i}^{c,r}$. Following the above derivations, one can observe that the multiplicative structure has been transformed into a more tractable summation form.

\section{Proposed Scheduling Algorithm}
    The reformulation in the previous section enables us to develop a simple algorithm that can achieve high scheduling performance. Specifically, by \eqref{JTSINR1} and \eqref{rate1}, we approximate problem \eqref{SRmax-original} by the following problem with a separable objective and constraints 
    \begin{subequations}\label{SRmax_overall}
        \begin{align}
                \max\nolimits_{\{b^{c,r}_{k}\}}  & \sum\nolimits_{(c,r)}\sum\nolimits_{k \in \mathcal{\hat{Q}}} \sum\nolimits_{m \in \mathcal{B}_k}  \tilde{f}_{m, k}^{c,r} \label{SRmax2} \\
                \text { s.t. } & b^{c,r}_{k}\in\{0,1\}, \forall c\in \mathcal{C}, r\in \mathcal{R}, k\in \mathcal{K}\\
                & \sum\nolimits_{(c,r)} \sum\nolimits_{m \in \mathcal{B}_j}\tilde{f}_{ m,j}^{c,r} \geq Q_{j},  \forall  j  \in \mathcal{Q}, \label{KNF}
        \end{align}
    \end{subequations}
    
    \noindent which remains a constrained non-convex problem. To tackle the constraints, we apply a \textit{min}-based penalty to the rate QoS constraints \eqref{KNF}, reformulating \eqref{SRmax_overall} as
     \begin{subequations}\label{minpenalty}
     \begin{align} 
        \max\nolimits_{\{b^{c,r}_{k}\}}& \ G \triangleq  \sum\nolimits_{(c,r)}\sum\nolimits_{k \in \mathcal{\hat{Q}}} \sum\nolimits_{m \in \mathcal{B}_k}  \tilde{f}_{m, k}^{c,r} \\ 
        &\!\!\!\!\!\!\!\!+ \rho\sum\nolimits_{k \in \mathcal{Q}} \min{\{\sum\nolimits_{(c,r)} \sum\nolimits_{m \in \mathcal{B}_j}\tilde{f}_{ m,j}^{c,r}, Q_{j}\}}, \notag\\
        \text { s.t. } & b^{c,r}_{k}\in\{0,1\}, \forall c\in \mathcal{C}, r\in \mathcal{R}, k\in \mathcal{K}
     \end{align}
     \end{subequations}
     where $\rho$ is a penalty parameter. For the above transformation, the following theorem holds true.
    \begin{Theorem} 
        There exists an $\rho^*>0$, when $\rho>\rho^* $, the local optimal solution of the above problem \eqref{minpenalty} is also the local optimal solution of the constraint problem \eqref{SRmax_overall} \cite{nocedal1999numerical}.
    \end{Theorem} 

    \begin{figure}[t]
		\begin{algorithm}[H]
				\caption{Proposed Scheduling Algorithm}
				\begin{algorithmic}[1] \label{algor1}
					\STATE Initialize $\{b_{k}^{c,r}\}$; 
					\FOR{each UE $t$}
					\FOR{each RBG $(c,r)$}
					\IF{$\text{gain}[c,r,t] > 0$}
					\STATE  $b_{t}^{c,r} = 1$; 
					\ELSE
					\STATE $b_{t}^{c,r} = 0$; %
					\ENDIF
					\ENDFOR
					\ENDFOR
					\STATE Repeat Steps 2-10 until maximum iterations reached. 
				\end{algorithmic}
		\end{algorithm}
	\end{figure}

    So far, problem \eqref{minpenalty} has been as a problem over binary variables. One way to solve it is to relax the binaries to continuous values within $[0, 1]$ and add a negative $\ell_2$-norm regularizer \cite{shao2019framework}. However, this approach will introduce more parameters to tune and may yield infeasible solutions after rounding. Thus, we choose to solve \eqref{minpenalty} directly by using the BCD method, as it ensures that the iterative updates produce a non-decreasing sequence of objective values and converge to a local optimum. Specifically, for each $b^{c,r}_{k}$, with other scheduling variables fixed, the gain of assigning UE $k$ to RBG $(c,r)$ is defined as
    \begin{align} \label{eq:gain1}
	\text{gain}[c,r,k] \triangleq G(b^{c,r}_{k}=1) - G(b^{c,r}_{k}=0).
    \end{align}
    If the gain is positive, we set $b^{c,r}_k = 1$; otherwise, $b^{c,r}_k = 0$. By applying this process iteratively to all variables, we obtain a computation-efficient scheduling algorithm, which is summarized in Algorithm~\ref{algor1}.

    \section{Simulation Results} \label{sec:simu}

    In the simulation, we consider a 3-cell network, each deploying one BS located at $[0, -300, 25]$, $[-1000, -300, 25]$, and $[-500, -1200, 25]$ (m). $K$ UEs are randomly located in the area $[-1400, 400]\times[-1400, -100]$ (m), each with height 1.5m and $N_r=4$ antennas.
    One-third of UEs are randomly selected and imposed with rate constraints.
    The spectrum consists of three CCs centered at 3.2, 3.5, and 3.8 GHz, each containing 624 subcarriers grouped into 13 RBGs of 48 subcarriers.  The noise spectral density is set to $-174$dBm/Hz and $P_m = 10 \text{dBm}$, while the channels are generated by Quadriga\cite{jaeckel2014quadriga}. For the BS-UE association, we apply the channel-strength strategy \cite{tanbourgi2014analysis}: each UE connects to the BS with the strongest large-scale gain, while secondary links for JT are added if the gap of their channel gain to the strongest one is within $\beta = 5$~dB, unless stated otherwise.
    Performance is evaluated with two metrics: the effective sum rate {\bf{(ESR)}}  $G_{esr}$, which takes $G$ defined in \eqref{minpenalty} with $\rho = 1$ and $\{\tilde f_{m,k}^{c,r}, \tilde f_{m,j}^{c,r}\}$ replaced by the true rates in \eqref{realrate} without approximation; and the  QoS satisfaction rate {\bf{$\mathrm{\textbf{Sat}}$}}, corresponding to the ratio of UEs whose QoS requirements are satisfied.

   Fig. \ref{Centralized Algorithm} shows the objective value of problem \eqref{minpenalty} achieved by the proposed scheduling algorithm across iterations, with $K=45$. Specifically, \( f_{{a}} \triangleq G/ (SCK) \), with $G$ expressed in \eqref{minpenalty} with $\rho =1$, corresponds to the averaged approximate ESR per RBG per UE, while \( f_{{t}} \triangleq G_{\text{esr}} / (SCK) \) is the attained averaged ESR without approximation. Fig. \ref{Centralized Algorithm} presents that the value of  \( f_{{a}} \) increases rapidly and gets stable within about 5 iterations, indicating the fast convergence of our proposed design. Moreover, the trend of \( f_t\) follows that of $f_a$ closely. The relative error between them is consistently below 3\% whenever the \( N_t \) is 32 or 64, which confirms the effectiveness of our proposed approximation in Sec. \ref{sec:reformu}.

	\begin{figure}[t] 
        \centering{\includegraphics[width=0.46\textwidth]{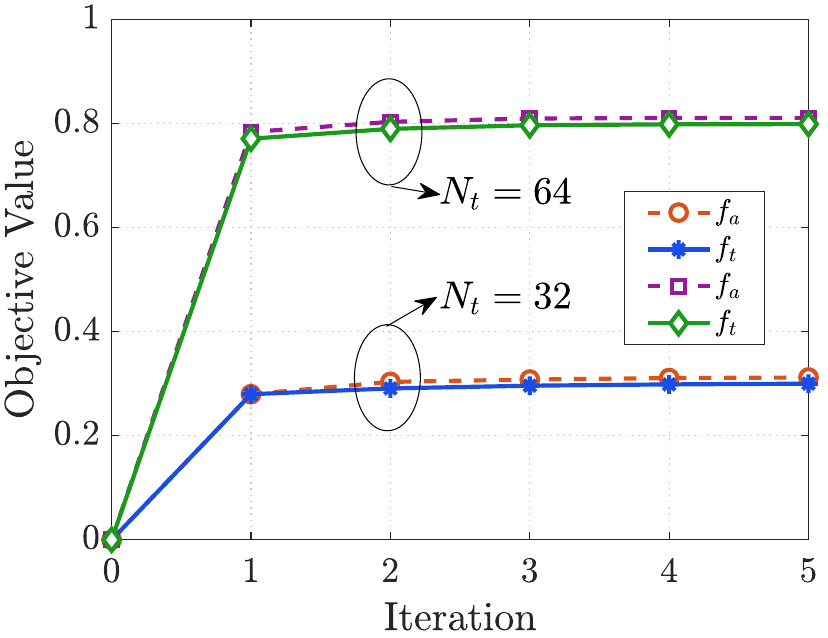}}	
        \caption{Convergence of the proposed scheduling algorithm. }
        \label{Centralized Algorithm}
    \end{figure}
    \begin{figure}[t] 
		\centering {\includegraphics[width=0.46\textwidth]{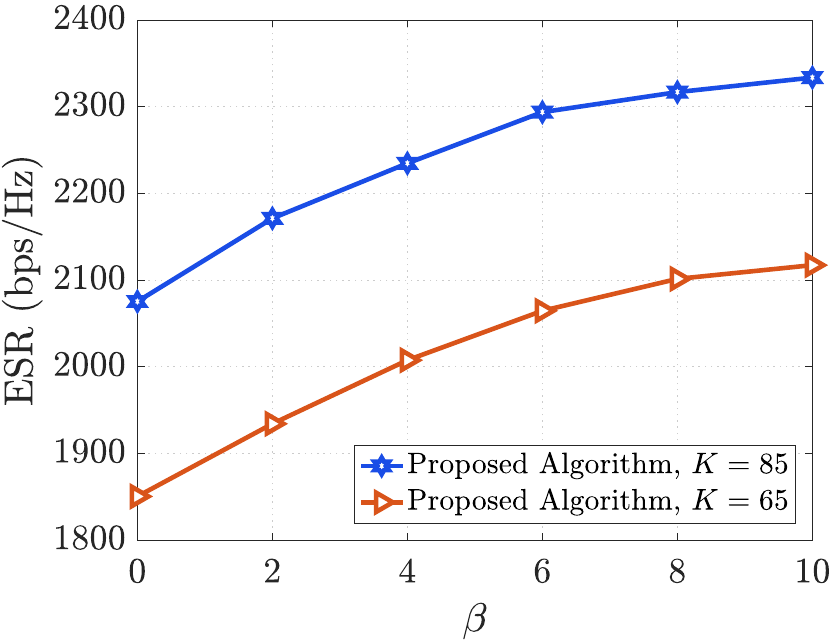}}	
		\caption{ESR versus the association threshold of JT-UEs,  $N_t = 64$.}
		\label{Centralized Algorithm-1}
	\end{figure} 
    \begin{figure}[t] 
        \centering{\includegraphics[width=0.46\textwidth]{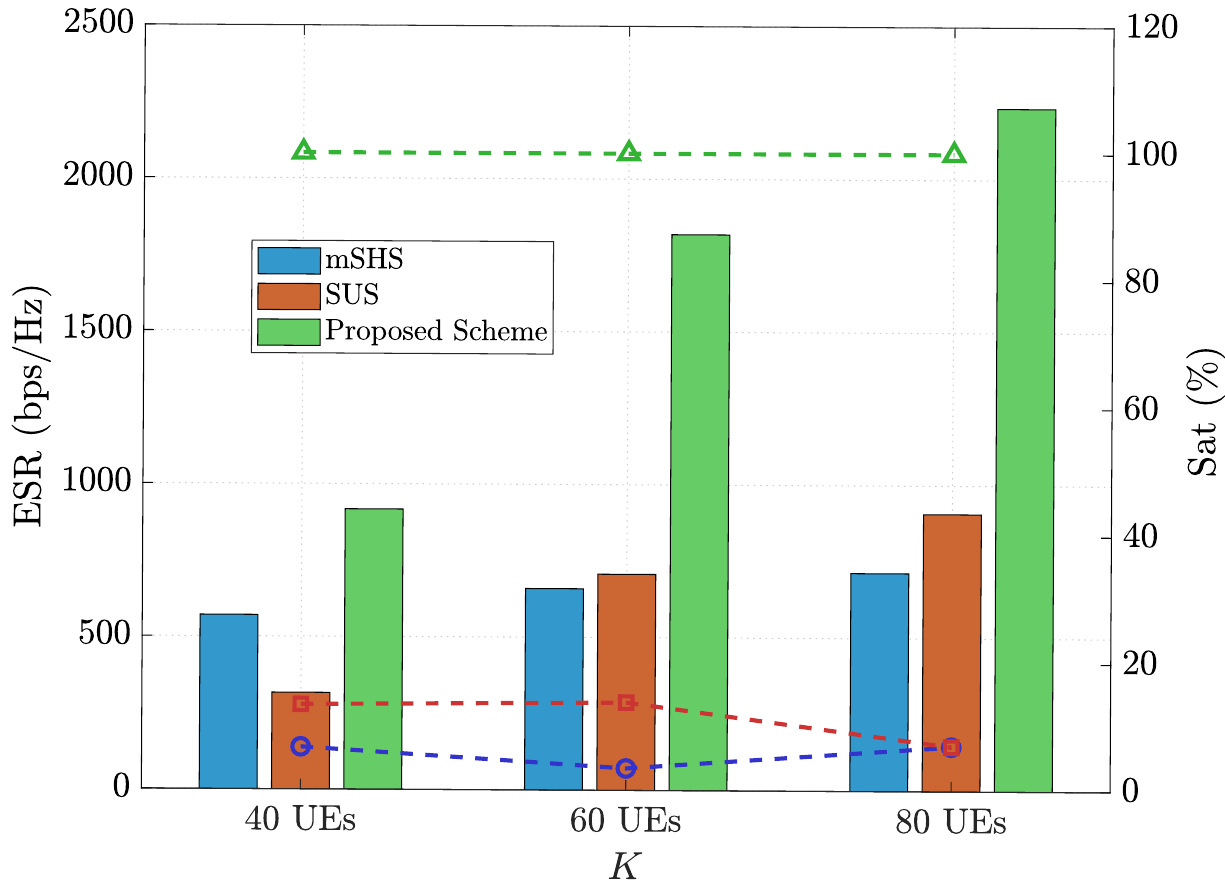}}
        \caption{ESR (bar charts) and $\mathrm{Sat}$ (dashed-line-connected symbols) achieved by different schemes versus different numbers of UEs $K$, $N_t = 64$.} 
        \label{comp}
    \end{figure}  

    Fig.~\ref{Centralized Algorithm-1} examines the impact of $\beta$ on the ESR achieved by the proposed scheme for two different UE configurations: $K=65$ and $K=85$. With a larger $\beta$, UEs are more likely to be served by more than one BS, and the ratio of JT-UEs will increase. At $\beta = 0\ \text{dB}$, every UE connects exclusively to a single BS that offers the strongest gain, and there is no JT-UE in the network.
    The ESR in Fig. \ref{Centralized Algorithm-1} is lowest at $\beta = 0\ \text{dB}$ due to the absence of JT, which results in high ICI. For both $K=65$ and $K=85$, we observe a similar trend where ESR increases with $\beta$, albeit with different absolute values. The curve for $K=85$ is consistently higher than that for $K=65$, indicating that having more UEs in the system leads to higher ESR across all values of $\beta$.
    As $\beta$ rises, more UEs become eligible for JT, which transforms the detrimental inter-cell interference into useful signals, and thus improves the throughput by exploiting spatial diversity and mitigating ICI effectively with our proposed scheduling. The most significant gains occur for $\beta \in [2, 6]\ \text{dB}$ for both configurations. Beyond $6\ \text{dB}$, further increases yield diminishing throughput improvement, as adding highly asymmetric links with weaker BSs contributes marginally. 

    In Fig. \ref{comp}, our proposed design is compared against two existing schemes: \textbf{SUS} \cite{yoo2006optimality} with EZF-BF and a modified SINR-based heuristic scheme (\textbf{mSHS}), which modifies SHS \cite{idrees2022throughput} by weighting the user selection according to individual QoS demands. One can observe that our proposed design achieves the highest ESR while maintaining 100\% Sat. In contrast, both mSHS and SUS exhibit significant performance loss. The mSHS, while improving upon SHS by introducing QoS-aware weighting, still suffers from low ESR and Sat. This is largely due to its conservative single-UE-per-RBG allocation, which restricts multi-user spatial multiplexing and leads to inefficient resource utilization. Additionally, though SUS achieves higher Sat and ESR than mSHS for $K = 60$ and $80$, its performance remains far below that of our proposed scheme. This is because SUS schedules UEs indiscriminately without effectively accounting for their actual QoS requirements. Overall, the proposed design brings substantial performance improvements, offering a highly competitive solution for multi-cell MU-MIMO scheduling.

    \section{Conclusion}

    In this work, we studied the cross-carrier cross-cell scheduling in multi-user MIMO networks, including both JT-UEs and NJT-UEs. To handle the challenging throughput maximization problem subject to rate constraints, we have proposed a novel reformulation scheme by integrating EZF-BF and massive MIMO asymptotic analysis. This effectively eliminated the intermediate BF variables,  achieving an approximate problem with a separable and tractable structure. For this reformulated problem, we developed an efficient penalty-based block coordinate descent scheduling algorithm. In the simulation, the effectiveness of our proposed reformulation and the impact of JT on the throughput were examined. Moreover, our proposed algorithm achieved distinct performance improvement over benchmark schemes, offering a competitive scheduling solution for next-G wireless networks.

    \bibliographystyle{IEEEtran}
	\bibliography{refs_msbf}
\end{document}